\documentclass[letterpaper, 10 pt, conference]{ieeeconf}  

\IEEEoverridecommandlockouts                              

\usepackage{graphics} 
\usepackage{epsfig} 
\usepackage{amsmath} 
\usepackage{amssymb}  
\usepackage{color,soul}
\usepackage{balance}

\newtheorem{theorem}{Theorem}
\newtheorem{lemma}{Lemma}

\newtheorem{definition}{Definition}

\title{\LARGE \bf
Exactly Optimal Quickest Change Detection of Markov Chains
}

\author{Jason J. Ford, Justin M. Kennedy, Caitlin Tompkins, Jasmin James, Aaron McFadyen
\thanks{J. J. Ford, J. M. Kennedy, and A. McFadyen are with the School of Electrical Engineering and Robotics, Queensland University of Technology, 2 George St, Brisbane QLD, 4000 Australia. 
C. Tompkins was with the School of Electrical Engineering and Robotics, Queensland University of Technology.
J. James is with the School of Mechanical \& Mining Engineering, University of Queensland, Brisbane QLD, 4072 Australia.
The authors acknowledge continued support from the Queensland University of Technology (QUT) through the Centre for Robotics.
        {\tt\footnotesize j2.ford@qut.edu.au, j12.kennedy@qut.edu.au, 
        caitlin.tompkins@gmail.com,
        jasmin.martin@uq.edu.au, aaron.mcfadyen@qut.edu.au}}%
}

\begin{document}

\maketitle
\thispagestyle{empty}
\pagestyle{empty}

\begin{abstract}
This paper establishes that an exactly optimal rule for Bayesian Quickest Change Detection (QCD) of Markov chains is a threshold test on the no change posterior.
We also provide a computationally efficient scalar filter for the no change posterior whose effort is independent of the dimension of the chains.
We establish that an (undesirable) weak practical super-martingale phenomenon can be exhibited by the no change posterior when the before and after chains are too close in a relative entropy rate sense. 
The proposed detector is examined in simulation studies. 
\end{abstract}

\section{Introduction}
Quickest change detection (QCD) problems consider the detection of a change in the statistical properties of an observed process and occur in a wide variety of applications including quality control \cite{nikiforov}, anomaly detection \cite{Vaswani}, statistics \cite{tartakovsky2019asymptotically}, target detection \cite{James2019} and fault detection \cite{nikiforov}. In such problems, a sequence of quantitative measurements are monitored to extract information of the current properties of an observed process with the desire to detect a change as quickly as possible after it occurs subject to a false alarm constraint \cite{Basseville}.

Several formulations for QCD problems exist which vary in assumptions around the change point and optimality criteria used for design.
Some of the earliest formulations, now termed Bayesian formulations, were developed by Shiryaev assuming that the change point is a random variable with a known geometric prior and observations are independent and identically distributed (i.i.d.) \cite{ShiryaevOS}.
Under these assumptions, Shiryaev was able to establish an optimal rule  comparing the posterior probability of the change event against a threshold. 

More recently the Bayesian formulation has been extended to encompass non-geometrically distributed change-times \cite{Tartakovsky05,KrishnamurthyPhaseTime} and dependent data (i.e., non-i.i.d. observations) \cite{Dayanik,Tartakovsky05,Fuh}. 
However, despite these various formulations appearing in the literature, establishing optimal detection rules for dependent data and arbitrary change-time distributions remains a challenging problem.
In  \cite{Tartakovsky05} the authors considered the general non-i.i.d. case, and  demonstrated asymptotic optimality for Shiryaev's rule, with similar shown in \cite{Fuh} for Hidden Markov models (HMMs) under some regularity conditions.
An $\epsilon$-optimal solution approach for joint quickest detection and isolation problem for Markov modulated process is examined in \cite{Dayanik}, which the authors point out include Markov chains   as special case.
Recently (exact) optimal Bayesian QCD for HMM results have been established \cite{Ford2022OptimalBayesianQuickest}.

Markov chains play a fundamental role in a  wide variety of fields including susceptible-infectious-recovered (SIR) and Susceptible-Infectious-Susceptible (SIS) epidemic models of infectious diseases \cite{SIRSIS2011}, models of rumour spreading \cite{Rumour2017}, social media anomaly detection \cite{SocialMedia2016} and air traffic management \cite{AaronMC}. 
Despite being a fundamental model class, Bayesian QCD of Markov chains has not yet been fully resolved with progress being limited to the early work that established optimality of a threshold rule dependent on the current state of the Markov chain \cite{yakir}, and the $\epsilon$-optimal HMM QCD convex stopping sets results of \cite{Dayanik}.
Further, neither of these results provide an  efficient computation method to implement their rules.

In this paper we present an exactly optimal solution for Bayesian QCD of Markov chains (which to date has only been shown $\epsilon$-optimally or with limiting dependencies). 
We also provide algorithms for efficient calculation of the optimal solution.
Finally, we  provide novel insights into numerical properties of the optimal solution when the before and after Markov chains are insufficiently separated and an undesirable phenomenon occurs.
Specifically, the key contributions of the paper are:
\begin{enumerate}
    \item Establishing an exactly optimal stopping rule for Bayesian QCD of Markov chains
    that is a simple threshold test of conditional posterior information. 
    \item Providing an efficient scalar filter solution to implement the stopping rule.
    \item Establishing that when the before and after change Markov chains are insufficiently separated, in the sense of relative entropy rate, a weak practical super-martingale phenomenon can emerge.
\end{enumerate}

The paper is structured as follows: In Section 
\ref{sec:problem} we pose the problem of Bayesian QCD of Markov chains, before presenting the optimal rule and a computationally efficient solution in Section \ref{sec:main}.
In Section \ref{sec:separated}, we investigate what occurs when the Markov chains are insufficiently separated and establish the existence of a weak practical super-martingale phenomenon.
In Section \ref{sec:results}, we present an illustrative example.
We present some brief conclusions in Section \ref{sec:con}.

\section{Problem Formulation}\label{sec:problem}
In this section we pose the problem of Bayesian QCD for Markov Chains. 
\subsection{State Process}
Let us first define a space $S \triangleq \{e_1, \dots, e_{N}\}$ where $e_i \in \mathbb{R}^{1 \times N}$ are indicator vectors with 1 in the $i$th element and zeros elsewhere, where $N$ is the number of possible values of the Markov chain.
For $k\ge0$, we consider a process $X_k \in S$ whose statistical properties change at some time $\nu\ge1$.
For $k<\nu$, $X_k \in S$ can be modelled as a first-order time-homogeneous Markov chain described by the \emph{before}-change transition probability matrix with elements $A_b^{i,j} \triangleq P(X_{k+1} =e_i | X_{k} =e_j)$ for $1\le i,j \le N$. 
For $k\ge\nu$, $X_k \in S$ can be modelled as a different first-order time-homogeneous Markov chain described by the \emph{after}-change transition probability matrix with elements $A_a^{i,j} \triangleq P(X_{k+1} =e_i | X_{k} =e_j)$ for $1\le i,j \le N$. 
For simplicity of presentation, we assume throughout that both $A_b$ and $A_a$ are aperiodic and irreducible.
We assume the initial distribution for $X_0$ is known as $p(X_0)$, and that the process $X_k$ is observed and let $X_{[0,k]} \triangleq \{X_0 , \ldots, X_k\}$ be shorthand for the chain sequence until time $k$.

\subsection{Probability Measure Space Construction}
We will now follow the construction of \cite{Ford2022OptimalBayesianQuickest} and introduce a probability measure space that allows us to formally state our Bayesian QCD for Markov chains problem. Let $\mathcal{F}_k=\sigma(X_{[0,k]})$ denote the filtration generated by $X_{[0,k]}$. 
We consider a probability measure space $(\Omega, \mathcal{F}, P_\nu)$ where $\Omega$ is sample space of sequences of $X_{[0,\infty]}$, $\sigma$-algebra $\mathcal{F} = \cup_{k=1}^\infty \mathcal{F}_k$ with the convention that $\mathcal{F}_0 = \{0,\Omega\}$, and $P_{\nu}$ is the probability measure constructed using Kolmogorov's extension on the following probability density function of the state sequence
\begin{align*}
  p_\nu(X_{[0,k]}) &= \Pi_{\ell=\nu+1}^{k} A_a^{\zeta(X_{\ell}),\zeta(X_{\ell-1})} \nonumber \\ & \times \Pi_{\ell=1}^{\min(\nu,k)} A_b^{\zeta(X_{\ell}),\zeta(X_{\ell-1})} p(X_0)  
\end{align*}
where $\zeta(e_i) \triangleq i$ returns the index of the non-zero element of an indicator vector $e_i$, and we define $\Pi_{\ell=\nu+1}^{k} A_a^{\zeta(X_{\ell}),\zeta(X_{\ell-1})} \triangleq 1$ if $k < \nu+1$, and $\Pi_{\ell=1}^{\min(\nu,k)} A_b^{\zeta(X_{\ell}),\zeta(X_{\ell-1})} =1$ if $\nu=1$.
We highlight that $P_a$ and $P_b$ correspond to the special case measures corresponding to transitions according to $A_a$ or $A_b$ at all time, respectively, and let $E_a$ and $_b$ denote the corresponding expectation operations.

It will later be useful to note that the relative entropy rate between two Markov chains with transition probability matrices $A_b$ and $A_a$ can be shown to be given by \cite{Xie2005} (if $A_b$ is irreducible):
\begin{equation*}
    \mathcal{R}(A_b|A_a)=\sum_{i=1}^N \sum_{j=1}^N 
    \tilde{a}_b^j
    A_b^{i,j} \log\left(\frac{A_b^{i,j} }{A_a^{i,j} }\right)
\end{equation*}
where $\tilde{a}_b$ is the invariant stationary distribution of $A_b^{i,j}$.

\subsection{Change Time Prior}
In the Bayesian QCD problem considered in this paper, the change time $\nu\ge 1$ is an unknown random variable having a prior distribution $\pi_k=P(\nu=k)$.
This allows us to construct an average measure $P_\pi(G)= \sum_{k=1}^\infty \pi_k (G) P_k(G)$ for all $G\in \mathcal{F}$ and we let $E_\pi$ denote the corresponding expectation operation. In this work we assume a prior geometry in nature in that $\pi_k=(1-\rho)^{k-1} \rho$, with $\rho \in (0,1)$, as introduced by Shiryeav \cite{Shiryaev}.

\subsection{Cost Formulation}
We can now state our QCD problem as seeking to quickly detect a change in the statistical properties of $X_k$ in the sense of designing a  stopping time $\tau \ge 1$ with respect to the filtration generated by $X_{[0,k]}$
that minimises the following cost (Bayes risk)
\begin{equation}
    J(\tau) \triangleq c E_\pi \left[(\tau-\nu)^+ \right] + P_\pi(\tau < \nu) \label{equ:cost}
\end{equation}
where $(\tau-\nu)^+ \triangleq \max(0, \tau-\nu)$ and $c$ is the penalty of each time step that alert is not declared after $\nu$.

By exploiting the recent QCD for HMM results of \cite{Ford2022OptimalBayesianQuickest} this paper extends the partial results of \cite{yakir} to establish an exactly optimal rule for QCD of Markov chains. 
Importantly, we show this optimal rule can be elegantly achieved through the efficient calculations of a scalar filter.
We also establish new insufficiently informative results in the Markov chains QCD setting inspired by 
recent {\it i.i.d.} QCD results \cite{FORD2020OnInformativenessofMeasurements}.

\section{Main Result}\label{sec:main}
In this section we present our main results for Bayesian QCD of Markov chains. We first present an augmented state representation before establishing the exactly optimal solution. We then provide an efficient calculation of the optimal solution.

\subsection{Exactly Optimal Solution}
Let us first define a new space $\bar{S} \triangleq \{\bar{e}_1, \dots, \bar{e}_{2N}\}$ where $\bar{e}_i \in \mathbb{R}^{1 \times 2N}$ are indicator vectors with 1 in the $i$th element and zeros elsewhere, and let us consider an augmented process $Z_k$.
Then for $k < \nu$, $Z_k \in \bar{S}$ is defined as 
\begin{equation*}
Z_k \triangleq \left[ 
    \begin{array}{c}
        X_k  \\ 0_{1 \times N}
    \end{array} 
    \right]
\end{equation*}
and for $k \ge \nu$, $Z_k \in \bar{S}$ is defined as 
\begin{equation*}
Z_k \triangleq \left[ 
    \begin{array}{c}
        0_{1 \times N} \\ X_k 
    \end{array}
    \right]
\end{equation*}
where $0_{1 \times N} $ is matrix of zeros of size $1 \times N$.
We note that $Z_k$ is only indirectly 
observed via the measured Markov chain $X_k$, and we later show $(X_k,Z_k)$ can be considered a hidden Markov model.

Let us consider a vector of conditional posterior probability having elements $\hat{Z}_k^i\triangleq P(Z_k=\bar{e}_i | X_{[0,k]})$ for all $i$ in $1 \le i \le 2N$ and define conditional posterior probability of before and after change as 
$\hat{M}_k^{b} \triangleq \sum_{i=1} ^ {N} \hat{Z}_k^i$ 
(noting that $\hat{M}_k^{b} + \sum_{i=N+1} ^ {2N} \hat{Z}_k^i = 1$.)

\begin{theorem}
\label{thm:optimalstoprule}
For the cost criteria \eqref{equ:cost} the exactly optimal stopping rule $\tau^*$ is given as 
\begin{equation}
    \tau^*\triangleq \{k \ge 1 : \hat{M}_k^{b} \le h \}
    \label{equ:stoprule}
\end{equation}
for some threshold value $h\in [0,1]$.
\end{theorem}

\begin{proof}
We note in the notation of \cite{Ford2022OptimalBayesianQuickest}, consider the before change and after change spaces $S_b = S$ and $S_a = S$, with corresponding transition probability matrices $A_b$ and $A_a$, and transition matrix $A_\nu = A_b$.  
Further consider their $y_k$ to be our $X_k$ in this paper, and set $b_b(y_k=e_i^b,e_j^b)=1$ if $i=j$ and 0 otherwise, and
$b_a(y_k=e_i^a,e_j^a)=1$ if $i=j$ and 0 otherwise. 
Then noting Theorem 1 of \cite{Ford2022OptimalBayesianQuickest} applies and the above theorem claim holds.
\end{proof}

This Theorem establishes that an exactly optimal stopping rule for Bayesian QCD of Markov chains is a simple threshold test on the conditional no change posterior information.
This result is stronger than previous QCD results for Markov chains which established an 
optimal stopping rule as a comparison of the no change posterior against a threshold having possible dependence on the current state of the Markov chain \cite{yakir}, as well as those that could be developed via the asymptotic HMM QCD results of \cite{Fuh} or the $\epsilon$-optimal HMM QCD convex stopping sets results of \cite{Dayanik}.

\subsection{Efficient Calculation of Optimal Solution}
We now investigate how to efficiently implement the optimal stopping rule through re-casting this calculation through the augmented hidden Markov model.
For that purpose, let us defined the emission matrix with elements $\mathcal{B}^{i,j} \triangleq P(X_k=e_i | Z_k = \bar{e}_j)$ for $1 \le i \le N$ and $1 \le j \le 2N$ and define a transition probability matrix with elements $A^{i,j} \triangleq P(Z_{k+1} =\bar{e}_i | Z_{k} =\bar{e}_j)$ for $1\le i,j \le 2N$.

\begin{lemma} \label{lemma:augmentedHMM}
The ($X_k$,$Z_k$) are the observation and state process for a hidden Markov model with an emission matrix $\mathcal{B} \in R^{N \times 2N}$ with elements
\begin{equation*}
    \mathcal{B}^{i,j}= \left\{ 
    \begin{array}{cc}
        1 & \mbox{if } i=j  \mbox{ or } i=N+j \\
        0 & \mbox{otherwise,}
    \end{array}
    \right.
\end{equation*}
or equivalently $\mathcal{B}=[I_{N \times N}\; I_{N \times N}]$,  and transition probability matrix
\begin{equation*}
    A = \left[
    \begin{array}{cc}
        (1-\rho) A_b  & 0_{N \times N} \\
        \rho A_b & A_a
    \end{array}
    \right]
\end{equation*}
where $I_{N \times N}$ is identity matrix of size $N \times N$ and $0_{N \times N} $ is the zero matrix of size $N \times N$.
\end{lemma}

\begin{proof}
We note in the notation of \cite{Ford2022OptimalBayesianQuickest}, as above, consider the before change and after change spaces $S_b = S$ and $S_a = S$, with corresponding transition probability matrices $A_b$ and $A_a$, and transition matrix $A_\nu=A_b$.  
Further consider their $y_k$ to be our $X_k$ in this paper, and set $b_b(y_k=e_i^b,e_j^b)=1$ if $i=j$ and 0 otherwise, and
$b_a(y_k=e_i^a,e_j^a)=1$ if $i=j$ and 0 otherwise. 
Then Lemma 2 of \cite{Ford2022OptimalBayesianQuickest} applies giving the lemma result here.
\end{proof}

The importance of Lemma 1 is that it establishes the conditional posteriors $Z_k^i$ and hence $\hat{M}_k^b$ can be efficiently calculated using a HMM filter as follows.
Let $\mathcal{B}^{i,.}$ denote the $i$th row of $\mathcal{B}$, and define the diagonal matrix $ {B} (X_k)\triangleq \mbox{diag}( \mathcal{B}^{\zeta(X_k),.}) \in \mathcal{R}^{2N \times 2N}$ and note this is a sparse diagonal matrix with 2 non-zero elements
\begin{equation*}
    {B} (X_k)^{i,i}
    = \left\{ 
    \begin{array}{cc}
        1 & \mbox{if } X_k=i  \mbox{ or } X_k=N+i \\
        0 & \mbox{otherwise}.
    \end{array}
    \right.
\end{equation*}

Then, for $k>0$, $\hat{Z}_k$ can be calculated using the HMM filter \cite{elliott1995}:
\begin{equation}
    \hat{Z}_{k} = N_k {B}(X_k) A \hat{Z}_{k-1}
    \label{equ:HMM}
    \end{equation}
where $N_k \triangleq \langle 1, {B} (X_k) A   \hat{Z}_{k}\rangle ^{-1}$ is a normalisation factor, and  $\hat{Z}_0=[\hat{X}_0',0_{1 \times N}']'$.

The following (perhaps) surprisingly efficient scalar filter implementation holds where the computation effort is independent of the size of the chains $N$.
\begin{lemma} \label{lem:efficient}
For $k>0$, the conditional no change posterior probability,
$\hat{M}_k^{b}$, can efficiently be calculated using the following scalar recursion
\begin{equation}
    \hat{M}_k^{b} = N_k (1-\rho) A_b^{\zeta(X_{k}),\zeta(X_{k-1})} \hat{M}_{k-1}^{b} 
    \label{equ:m1hat}
\end{equation}
where $ \hat{M}_0^{b}=1$ and we can calculate the normalisation factor as
\begin{align*}
 N_k^{-1} = &A_a^{\zeta(X_{k}),\zeta(X_{k-1})} \\&+ \hat{M}_{k-1}^{b} \left[ A_b^{\zeta(X_{k}),\zeta(X_{k-1})}-  A_a^{\zeta(X_{k}),\zeta(X_{k-1})} \right].
\end{align*}
\end{lemma}

\begin{proof}
First note that due to the dependence of $B(X_k)$ on $X_k$, $\hat{Z}_k$ is sparse in the sense that for each $k>0$ only the $\zeta(X_{k})$ and $N+\zeta(X_{k})$ elements are non-zero and hence $\hat{M}_k^{b} = \hat{Z}_k^{\zeta(X_{k})}$.
Therefore it follows by considering the location of non-zero elements of $\hat{Z}_k$ at times $k$ and $k-1$ in update step of \eqref{equ:HMM} that we can write $\hat{M}_k^{b}$ as \eqref{equ:m1hat}.
The expression for $N_k$ follows from noting $\hat{Z}_k^\zeta(X_{k})=1-\hat{Z}_k^{N+\zeta(X_{k})}$ and algebraic re-arrangement (via similar steps to those used in \cite[Lemma 1]{FORD2020OnInformativenessofMeasurements}). 
\end{proof}
Lemma \ref{lem:efficient} provides insights into the posterior filter computational structure and facilitates our analysis in the next section into what happens when chains are too close in a statistical sense.

\section{Insufficiently Separated Markov Chains}
\label{sec:separated}
In this section we will investigate the behaviour of our optimal rule's test statistic $\hat{M}_k^{b}$ in certain situations.
For that purpose, let us introduce the shorthand
$M_k = N_k (1-\rho) A_b^{\zeta(X_{k}),\zeta(X_{k-1})}$ which allows us to write the posterior probability update at time $k$ as $\log(\hat{M}_{k}^{b}) = \log(M_k) + \log(\hat{M}_{k-1}^{b} )$, and establish the following bound on $\log(M_k)$.
\begin{lemma} \label{lem:bound}
Assume $A_b$ and $A_a$ have unique stationary distributions and that the initial distribution $p(X_0)$ is the stationary distribution of $A_b$.
Then for any $\delta > 0$, there is a $h_\delta>0$ such that for any 
$\hat{M}_k^{b}< h _\delta$ we have
\begin{equation*}
    E_\pi  \left[  \log(M_k) \Big|\hat{M}_{k-1}^{b} \right] < \log(1-\rho) + \mathcal{R}( A_b | A_a )  + \delta
\end{equation*}
for sufficiently large $k$.
\end{lemma}

\begin{proof} This proof approach  is similar to the proof \cite[Lemma 2]{FORD2020OnInformativenessofMeasurements}, with the most significant difference being this result involves the relative entropy rate between chains (rather than relative entropy between measurement densities). 
We define 
\begin{equation*}
    \gamma_k=\log\left(A_a^{\zeta(X_{k}),\zeta(X_{k-1})} \right)-\log(N_k).
\end{equation*}
We then note we can write 
\begin{align*}
    E_\pi  \left[  \log(N_k) \Big|\hat{M}_{k-1}^{b} \right]
 = &- E_\pi  \left[  \log\left(A_a^{\zeta(X_{k}),\zeta(X_{k-1})} \right) \Big|\hat{M}_{k-1}^{b} \right] \\ & + E_\pi  \left[  \gamma_k \Big|\hat{M}_{k-1}^{b} \right].
\end{align*}
Hence we can write
\begin{align}
    &E_\pi  \left[   \log(M_k) \Big|\hat{M}_{k-1}^{b} \right] \nonumber \\
    &= \log(1-\rho) + E_\pi  \left[  \log\left(A_a^{\zeta(X_{k}),\zeta(X_{k-1})} \right) \Big|\hat{M}_{k-1}^{b} \right] \nonumber \\ &\quad - E_\pi  \left[  \log\left(A_a^{\zeta(X_{k}),\zeta(X_{k-1})} \right) \Big|\hat{M}_{k-1}^{b} \right]  \nonumber \\ &\quad + E_\pi  \left[  \gamma_k \Big|\hat{M}_{k-1}^{b} \right] 
    \nonumber \\
    &= \log(1-\rho) + E_\pi  \left[  \log\left( \frac{A_b^{\zeta(X_{k}),\zeta(X_{k-1})}}{A_a^{\zeta(X_{k}),\zeta(X_{k-1})}} \right) \Bigg|\hat{M}_{k-1}^{b} \right] \nonumber \\ &\quad + E_\pi  \left[  \gamma_k \Big|\hat{M}_{k-1}^{b} \right].
\label{equ:bound}
\end{align}

It will soon be useful to note that
\begin{align*}
    P_\pi( X_k \in S^b, X_{k-1} \in S^b  | \hat{M}_{k-1}^{b}) &= (1-\rho) \hat{M}_{k-1}^{b}, \\
    P_\pi( X_k \in S^a, X_{k-1} \in S^b  | \hat{M}_{k-1}^{b}) &= \rho \hat{M}_{k-1}^{b}, 
    \\
    P_\pi( X_k \in S^a, X_{k-1} \in S^a  | \hat{M}_{k-1}^{b}) &= 1 - \hat{M}_{k-1}^{b},
\end{align*}
that under the lemma assumptions on $p(X_0)$ we can write
\begin{align*}
    &P_\pi( X_k=e_i,X_{k-1}=e_j| X_k \in S^b, X_{k-1} \in S^b  , \hat{M}_{k-1}^{b}) \\&\quad= A_b^{i,j} P_b (X_{k-1} = e_j) \quad \mbox{and}\\
    &P_\pi( X_k=e_i,X_{k-1}=e_j| X_k \in S^a, X_{k-1} \in S^b  , \hat{M}_{k-1}^{b}) \\&\quad= A_b^{i,j} P_b (X_{k-1} = e_j), 
\end{align*}
and that for sufficient large $k$ we can write 
\begin{align*}
    &P_\pi( X_k=e_i,X_{k-1}=e_j| X_k \in S^a, X_{k-1} \in S^a  , \hat{M}_{k-1}^{b}) \\&\quad= A_a^{i,j} P_a (X_{k-1} = e_j) 
\end{align*}
where $P_b( X_{k-1} = e_j )$ and $P_a( X_{k-1} = e_j )$ are the stationary distributions of the before and after change models, respectively.

Then application of the law of total probability and Bayes' rule gives that, for sufficient large $k$, we can write
\begin{align}
     & P_\pi( X_k=e_i,X_{k-1}=e_j | \hat{M}_{k-1}^{b}) = \nonumber \\ & P_\pi( X_k=e_i,X_{k-1}=e_j, X_k \in S^b, X_{k-1} \in S^b  | \hat{M}_{k-1}^{b}) \nonumber \\ & + P_\pi( X_k=e_i,X_{k-1}=e_j, X_k \in S^a, X_{k-1} \in S^b  | \hat{M}_{k-1}^{b}) \nonumber \\ & + P_\pi( X_k=e_i,X_{k-1}=e_j, X_k \in S^a, X_{k-1} \in S^a | \hat{M}_{k-1}^{b}) \nonumber \\
     &= P_b (X_{k-1}=e_j) A_b^{i,j} \rho  \hat{M}_{k-1}^{b} \nonumber \\&\quad+ P_b (X_{k-1}=e_j) A_b^{i,j} (1-\rho)  \hat{M}_{k-1}^{b} \nonumber \\ &\quad+ P_a (X_{k-1}=e_j) A_a^{i,j} (1- \hat{M}_{k-1}^{b} ) \nonumber \\
     &= P_b (X_{k-1}=e_j) A_b^{i,j}   \hat{M}_{k-1}^{b} \nonumber \\ &\quad+ P_a (X_{k-1}=e_j) A_a^{i,j} (1- \hat{M}_{k-1}^{b} ). \label{equ:pexpand}
\end{align}

For sufficiently large $k$, we can now expand the second term of \eqref{equ:bound} as
\begin{align}
    &E_\pi  \left[  \log\left( \frac{A_b^{\zeta(X_{k}),\zeta(X_{k-1})}}{A_a^{\zeta(X_{k}),\zeta(X_{k-1})}} \right) \Bigg|\hat{M}_{k-1}^{b} \right] \nonumber \\ 
    &=\sum_{i=1}^N  \sum_{j=1}^N \log\left(\frac{A_b^{i,j}}{A_a^{i,j}} \right) P_\pi( X_k=e_i,X_{k-1}=e_j | \hat{M}_{k-1}^{b}) \nonumber \\
    &= \hat{M}_{k-1}^{b} \sum_{i=1}^N  \sum_{j=1}^N  P_b (X_{k-1}=e_j) A_b^{i,j}  \log\left( \frac{A_b^{i,j}}{A_a^{i,j}} \right) \nonumber \\ & \quad+  (1- \hat{M}_{k-1}^{b} ) \sum_{i=1}^N  \sum_{j=1}^N   P_a (X_{k-1}=e_j) A_a^{i,j}  \log\left(\frac{A_b^{i,j}}{A_a^{i,j}} \right) \nonumber \\
    &= \hat{M}_{k-1}^{b} \sum_{i=1}^N  \sum_{j=1}^N  P_b (X_{k-1}=e_j) A_b^{i,j}  \log\left(\frac{A_b^{i,j}}{A_a^{i,j}} \right) \nonumber \\ & \quad-  (1- \hat{M}_{k-1}^{b} ) \sum_{i=1}^N  \sum_{j=1}^N   P_a (X_{k-1}=e_j) A_a^{i,j}  \log\left(\frac{A_a^{i,j}}{A_b^{i,j}} \right) \nonumber \\
    & = \hat{M}_{k-1}^{b}  \mathcal{R}( A_b | A_a ) - (1-\hat{M}_{k-1}^{b}) \mathcal{R}( A_a | A_b) \nonumber \\ & \le \mathcal{R}( A_b | A_a ) 
    \label{equ:d}
\end{align}
where the first line follows from the definition of expectation operation, the 2nd line follows from application of \eqref{equ:pexpand}, the second last line follows from the definition of relative entropy rate between two chains and the last line follows because $\hat{M}_{k-1}^{b}\le 1$ and relative entropy rates such as $ \mathcal{R}( A_a | A_b)$ are non-negative.

The Lemma statement then follows from \eqref{equ:bound} and \eqref{equ:d} by noting that for any $\delta<0$ there is a $h_\delta>0$ such that for all $\hat{M}_{k-1} < h_\delta$ we have $E_\pi  \left[  \gamma_k \Big|\hat{M}_{k-1}^{b} \right]< \delta$ as $\log$ is monotonically increasing and the elements $A_b^{i,j}$ and $A_a^{i,j}$ are bounded in size (no greater than 1).
\end{proof}
 
This Lemma \ref{lem:bound} bound seems strikingly similar to the bounding result of Lemma 2 in \cite{FORD2020OnInformativenessofMeasurements} with relative entropy rate between chains replacing the role of relative entropy between measurement densities in that result.

The lemma's assumptions on  $p(X_0)$ being the stationary distribution of $A_b$ and the requirement for sufficiently largely $k$ have been included to simplify analysis
rather than being fundamental to the bounding mechanism.
We would expect similar bounded behaviour under relaxation of these two assumptions.
The lemma's assumption of the existence of stationary distributions is ensured under this paper's standing assumption of aperiodic and irreducible chains \cite[Ch. 4]{Cover2006}.
Whilst such ergodic chains are a large and useful class, the lemma does exclude non-ergodic chains such as those with transient states, absorbing states or chains that exhibit periodic behaviours. 

We now investigate a phenomenon that occurs when the Markov chains $A_b$ and $A_a$ are too close in the sense of having small relative entropy rate $\mathcal{R}(A_b|A_a)$ and are unable to overcome the change event's geometric prior information. 
For this purpose, consider the following concept of a weak practical super-martingale \cite{FORD2020OnInformativenessofMeasurements}:

\begin{definition} \label{def:wpsm}
(Weak Practical Super-martingale). If for any arbitrarily small $\delta_p>0$ there exists a $h_s>0$ such that if $M^b_k< h_s$ then 
\begin{align*}
    P_\pi\left(\mbox{for all } n \ge k, E_\pi [ \log(\hat{M}^b_{n+1}) | \log(\hat{M}^b_n) ]  <  \log(\hat{M}^b_n) \right) &  \\  & \hspace{-5em} > 1-\delta_p
\end{align*}
and the log of the no change posterior $\log(M^b_k)$ is called a weak practical super-martingale.
\end{definition}

The following theorem now holds.
\begin{theorem} \label{thm:informativeness}
Assume $A_b$ and $A_a$ have unique stationary distributions and that the initial distribution $p(X_0)$ is the stationary distribution of $A_b$.
If the chains are insufficiently separated in the sense that
\begin{equation}
   \mathcal{R}( A_b | A_a ) < \log(1/(1-\rho))
   \label{eq:informativebound}
\end{equation}
then for sufficiently large $k$ the log of no change posterior, $\log(\hat{M}_k^{b})$, is a weak practical super-martingale in the sense of Definition \ref{def:wpsm}.
\end{theorem}

\begin{proof}
Given result of Lemma \ref{lem:bound}, the theorem claim follows using the same proof steps as \cite[Thm. 4.]{FORD2020OnInformativenessofMeasurements}.
\end{proof}

The significance of this theorem is that unless the relative entropy rate between $A_b$ and $A_a$ is sufficiently large then the no change posterior $\log(\hat{M}_k^{b})$ is a weak practical super-martingale and hence there is a posterior interval trap $\hat{M}_k^{b}<h_s$ where the no change posterior $\hat{M}_k^{b}$ becomes increasingly confident that a change has occurred even when it has not occurred. 
That the test statistic can exhibit such behaviour if the before and after chain models are close is problematic, and can be interpreted as meaning that the models are insufficiently different to overcome the geometric prior.
The potential for this behaviour is an important design consideration. One practical remedy to avoid the super-martingale phenomenon might be to (artificially) reduce the value assumed for the geometric prior by a sufficient amount to ensure that the theorem condition no longer holds, and so that the no change posterior $\log(\hat{M}_k^{b})$ test statistic behaviour is a better indication of change status.

\section{Simulation study} \label{sec:results}
In this section we first illustrate the performance of our proposed optimal stopping rule in simulation example before examining the weak practical super-martingale phenomenon in some detail.

\subsection{Illustrative example} \label{sec:results:numerical}
Let us consider a three state Markov chain $X_k \in S = \{e_1, e_2, e_3\}$ with before and after transition probability matrices:
\begin{equation*}
    A_b = 
    \begin{bmatrix}
    0.99 & 0.005 & 0.005 \\ 0.005 & 0.99 & 0.005 \\ 0.005 & 0.005 & 0.99
    \end{bmatrix}
    \textrm{and} 
    \ A_a = 
    \begin{bmatrix}
    0.8 & 0.1 & 0.1 \\ 0.1 & 0.8 & 0.1 \\ 0.1 & 0.1 & 0.8 
    \end{bmatrix}.
\end{equation*}
The change event is assumed to have geometric prior $\rho = 0.005$.

Following Lemma~\ref{lemma:augmentedHMM}, we can consider a six state augmented process $Z_k$ with transition probability matrix
\begin{equation*}
    A = \begin{bmatrix}
        0.98505 & 0.004975 & 0.004975 & 0 & 0 & 0 \\
        0.004975 & 0.98505 & 0.004975 & 0 & 0 & 0 \\
        0.004975 & 0.004975 & 0.98505 & 0 & 0 & 0 \\
        0.00495 & 0.000025 & 0.000025 & 0.8 & 0.1 & 0.1 \\
        0.000025 & 0.00495 & 0.000025 & 0.1 & 0.8 & 0.1 \\
        0.000025 & 0.000025 & 0.00495 & 0.1 & 0.1 & 0.8
    \end{bmatrix} .
\end{equation*}

The top and middle sub figures of Figure~\ref{fig:toyproblem:states} shows a simulated example of the state of the augmented process $Z_k$ and Markov chain $X_k$, respectively.  
The change in the statistical properties of the measured process $X_k$ after change point $\nu=1000$ is visually apparent.
As shown in the bottom sub figure of Figure~\ref{fig:toyproblem:states}, our optimal stopping rule \eqref{equ:stoprule} is able to alert of the change when test statistic $\hat{M}_k^b$ crosses the alert threshold (say) $h=0.4$ at $k=1027$.

\begin{figure}
\centering
\includegraphics[scale=0.6,trim={0.0cm 0cm 0cm 0.0cm},clip]{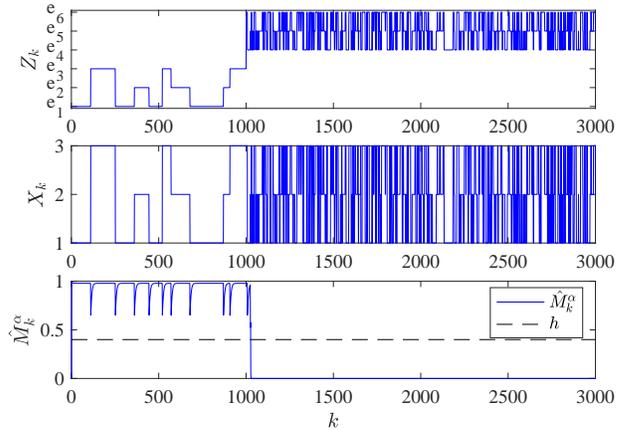}
\caption{Illustrative example of the augmented state process across two Markov chain each of three states. The top figure shows the current state of the augmented HMM. The middle figure is the measurements of the three states. The bottom figure is the test statistic with detection threshold marked as the dashed line.}
\label{fig:toyproblem:states}
\end{figure}

\subsection{Insufficiently Separated Markov Chains}
We now briefly investigate insufficiently separated Markov chains through a parametric study.
Consider a symmetric two state Markov chain with before and after change transition probability matrices:
\begin{equation*}
    A_b = \begin{bmatrix}
    0.99 & 0.01 \\ 0.01 & 0.99
    \end{bmatrix}, \quad \textrm{and} \quad
    A_a = \begin{bmatrix}
    a & (1-a) \\ (1-a) & a
    \end{bmatrix},
\end{equation*}
where we explore $a \in [0.84, 0.99]$ and assume fixed geometric prior $\rho = 0.005$.

We perform a Monte Carlo simulation of $1000$ trials for each $a \in [0.84, 0.99]$ for $5000$ time steps of the before change model (i.e. no change event occurs).

Figure~\ref{fig:informativeness:frequencyposterior} is the frequency of the no change posterior $\hat{M}_k^b$ larger than threshold of $h=0.001$ at $k=5000$ over $1000$ trials.
This illustrates that under the conditions of Theorem~\ref{thm:informativeness}, which provides that the weak practical super-martingale phenomenon occurs for $a \in (0.977, 0.99)$, the test statistic $\hat{M}_k^b$ becomes increasingly confident a change event has occurred, even though it has not.

\begin{figure}
    \centering
    \includegraphics[scale=0.6,trim={0.0cm 0cm 0cm 0.0cm},clip]{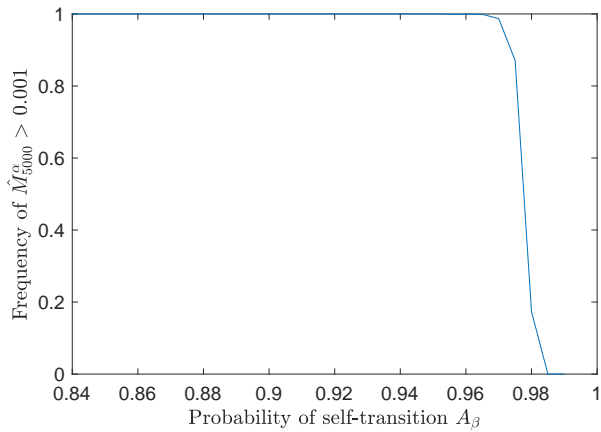}
    \caption{Monte Carlo simulation of no change scenario with two symmetric Markov chains. Frequency of the no change posterior probability $\hat{M}_k^b$ at $k=5000$ larger than threshold of $h=0.001$ over $1000$ trials.}
    \label{fig:informativeness:frequencyposterior}
\end{figure}

Figure~\ref{fig:informativeness:realisations} illustrates two realisations with the self-transition of $A_a$ exhibiting (dashed red line) and not exhibiting (solid blue line) the weak practical super-martingale phenomenon.
Where $a=0.985$ inside the conditions of Theorem~\ref{thm:informativeness}, the test statistic, or the posterior probability of the no change scenario, $\hat{M}_k^b$ becomes more confident a change event has occurred, even though it has not.

\begin{figure}
    \centering
    \includegraphics[scale=0.6,trim={0.0cm 0cm 0cm 0.0cm},clip]{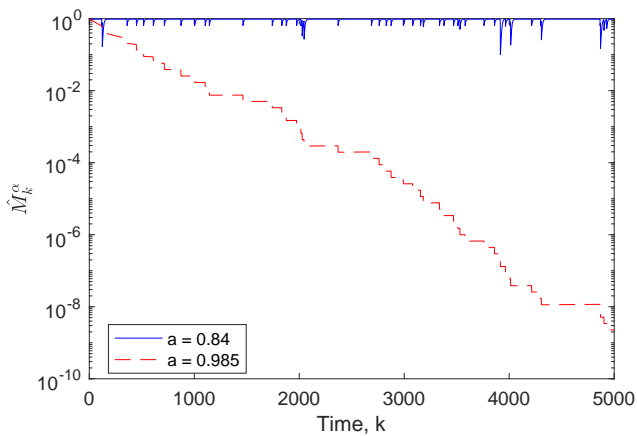}
    \caption{Two realisations from the Monte Carlo simulation of the no change scenario with two symmetric Markov chains. No change posterior probability with self-transition of $A_a$ outside the range $a=0.84$ (solid blue line), and exhibiting the weak practical super-martingale phenomenon $a=0.985$ (dashed red line).}
    \label{fig:informativeness:realisations}
\end{figure}

\section{Conclusions} \label{sec:con}
This  paper  established  an  exactly optimal  rule  for Bayesian Quickest Change Detection 
of Markov chains involving  a  threshold  test  on  the  no change  posterior. We provided an efficient computational approach. 
We also established an (undesirable) weak practical super-martingale phenomenon can be exhibited by the no change posterior when the before and after chains are too close in a statistical sense.  The potential for this phenomenon is an important design consideration.

\addtolength{\textheight}{-12cm}   




\balance

\bibliographystyle{IEEEtran}
\bibliography{IEEEabrv,ref}

\begin{thebibliography}{10}
\providecommand{\url}[1]{#1}
\csname url@samestyle\endcsname
\providecommand{\newblock}{\relax}
\providecommand{\bibinfo}[2]{#2}
\providecommand{\BIBentrySTDinterwordspacing}{\spaceskip=0pt\relax}
\providecommand{\BIBentryALTinterwordstretchfactor}{4}
\providecommand{\BIBentryALTinterwordspacing}{\spaceskip=\fontdimen2\font plus
\BIBentryALTinterwordstretchfactor\fontdimen3\font minus
  \fontdimen4\font\relax}
\providecommand{\BIBforeignlanguage}[2]{{%
\expandafter\ifx\csname l@#1\endcsname\relax
\typeout{** WARNING: IEEEtran.bst: No hyphenation pattern has been}%
\typeout{** loaded for the language `#1'. Using the pattern for}%
\typeout{** the default language instead.}%
\else
\language=\csname l@#1\endcsname
\fi
#2}}
\providecommand{\BIBdecl}{\relax}
\BIBdecl

\bibitem{nikiforov}
M.~Basseville and I.~Nikiforov, \emph{Detection of Abrupt Change Theory and
  Application}.\hskip 1em plus 0.5em minus 0.4em\relax Prentice-Hall, 04 1993,
  vol.~15.

\bibitem{Vaswani}
N.~{Vaswani}, ``Additive change detection in nonlinear systems with unknown
  change parameters,'' \emph{IEEE Transactions on Signal Processing}, vol.~55,
  no.~3, pp. 859--872, 2007.

\bibitem{tartakovsky2019asymptotically}
A.~G. Tartakovsky, ``Asymptotically optimal quickest change detection in
  multistream data—part 1: General stochastic models,'' \emph{Methodology and
  Computing in Applied Probability}, vol.~21, no.~4, pp. 1303--1336, 2019.

\bibitem{James2019}
J.~James, J.~J. Ford, and T.~L. Molloy, ``Quickest detection of intermittent
  signals with application to vision-based aircraft detection,'' \emph{IEEE
  Transactions on Control Systems Technology}, vol.~27, no.~6, pp. 2703--2710,
  2019.

\bibitem{Basseville}
M.~Basseville, ``{Detecting changes in signals and systems - A survey},''
  \emph{Automatica}, vol.~24, no.~3, pp. 309--326, May 1988.

\bibitem{ShiryaevOS}
A.~N. Shiryaev, \emph{{Optimal Stopping Rules}}.\hskip 1em plus 0.5em minus
  0.4em\relax Springer-Verlag Berlin Heidelberg, 2008, vol.~8.

\bibitem{Tartakovsky05}
A.~Tartakovsky and V.~Veeravalli, ``\BIBforeignlanguage{English (US)}{General
  asymptotic {B}ayesian theory of quickest change detection},''
  \emph{\BIBforeignlanguage{English (US)}{Theory of Probability and its
  Applications}}, vol.~49, no.~3, pp. 458--497, 2005.

\bibitem{KrishnamurthyPhaseTime}
V.~Krishnamurthy, ``Bayesian sequential detection with phase-distributed change
  time and nonlinear penalty—a pomdp lattice programming approach,''
  \emph{IEEE Transactions on Information Theory}, vol.~57, no.~10, pp.
  7096--7124, 2011.

\bibitem{Dayanik}
S.~Dayanik and C.~Goulding, ``Sequential detection and identification of a
  change in the distribution of a {M}arkov-modulated random sequence,''
  \emph{IEEE Transactions on Information Theory}, vol.~55, no.~7, pp.
  3323--3345, July 2009.

\bibitem{Fuh}
C.~{Fuh} and A.~G. {Tartakovsky}, ``Asymptotic {B}ayesian theory of quickest
  change detection for hidden {M}arkov models,'' \emph{IEEE Transactions on
  Information Theory}, vol.~65, no.~1, pp. 511--529, Jan 2019.

\bibitem{Ford2022OptimalBayesianQuickest}
\BIBentryALTinterwordspacing
J.~J. Ford, J.~James, and T.~L. Molloy, ``Exactly optimal {B}ayesian quickest
  change detection for hidden {M}arkov models,'' 2021. [Online]. Available:
  \url{https://arxiv.org/abs/2009.00150}
\BIBentrySTDinterwordspacing

\bibitem{SIRSIS2011}
R.~Yaesoubi and T.~Cohen, ``Generalized {M}arkov models of infectious disease
  spread: A novel framework for developing dynamic health policies,''
  \emph{European Journal of Operational Research}, vol. 215, no.~3, pp.
  679--687, 2011.

\bibitem{Rumour2017}
G.~Ferraz~de Arruda, F.~Aparecido~Rodrigues, P.~Martín~Rodríguez, E.~Cozzo,
  and Y.~Moreno, ``{A general {M}arkov chain approach for disease and rumour
  spreading in complex networks},'' \emph{Journal of Complex Networks}, vol.~6,
  no.~2, pp. 215--242, 08 2017.

\bibitem{SocialMedia2016}
R.~Yu, H.~Qiu, Z.~Wen, C.~Lin, and Y.~Liu, ``A survey on social media anomaly
  detection,'' \emph{ACM SIGKDD Explorations Newsletter}, vol.~18, no.~1, pp.
  1--14, 2016.

\bibitem{AaronMC}
L.~Faulkner and A.~McFadyen, ``Air traffic configuration modelling and dynamic
  airspace allocation using discrete-time {M}arkov chains,'' in \emph{2019 IEEE
  Intelligent Transportation Systems Conference (ITSC)}, 2019, pp. 4483--4488.

\bibitem{yakir}
B.~Yakir, \emph{Optimal detection of a change in distribution when the
  observations form a Markov chain with a finite state space}, ser. Lecture
  Notes--Monograph Series.\hskip 1em plus 0.5em minus 0.4em\relax Hayward, CA:
  Institute of Mathematical Statistics, 1994, vol. Volume 23, pp. 346--358.

\bibitem{Xie2005}
L.~Xie, V.~Ugrinovskii, and I.~Petersen, ``Probabilistic distances between
  finite-state finite-alphabet hidden {M}arkov models,'' \emph{IEEE
  Transactions on Automatic Control}, vol.~50, no.~4, pp. 505--511, 2005.

\bibitem{Shiryaev}
A.~N. Shiryaev, ``On optimum methods in quickest detection problems,''
  \emph{Theory of Probability \& Its Applications}, vol.~8, no.~1, pp. 22--46,
  1963.

\bibitem{FORD2020OnInformativenessofMeasurements}
J.~J. Ford, J.~James, and T.~L. Molloy, ``On the informativeness of
  measurements in {S}hiryaev’s {B}ayesian quickest change detection,''
  \emph{Automatica}, vol. 111, p. 108645, 2020.

\bibitem{elliott1995}
R.~Elliott, L.~Aggoun, and J.~Moore, \emph{Hidden {M}arkov Models: Estimation
  and Control}.\hskip 1em plus 0.5em minus 0.4em\relax Springer-Verlag, 1995.

\bibitem{Cover2006}
T.~M. Cover and J.~A. Thomas, \emph{Elements of Information Theory 2nd Edition
  (Wiley Series in Telecommunications and Signal Processing)}.\hskip 1em plus
  0.5em minus 0.4em\relax Wiley-Interscience, July 2006.

\end{thebibliography}

\end{document}